\documentclass[journal, twocolumn]{IEEEtran}
\usepackage{amsmath}
\usepackage{amssymb}
\usepackage{enumerate}
\usepackage{graphicx}
\usepackage{epstopdf}
\usepackage{subfigure}
\usepackage{cite}
\usepackage{laterref}
\usepackage[mathscr]{euscript}
\usepackage{algorithm}
\usepackage{algpseudocode}
\DeclareGraphicsExtensions{.eps}
\newtheorem{theorem}{Theorem}[section]
\newtheorem{lemma}[theorem]{Lemma}

\usepackage{color}
\usepackage[usenames,dvipsnames]{xcolor}
\newcommand*{\QEDA}{\hfill\ensuremath{\blacksquare}}

\usepackage{tikz}
\usetikzlibrary{shapes,arrows}

\newenvironment{proof}[1][Proof]{\begin{trivlist}
\item[\hskip \labelsep {\bfseries #1}]}{\end{trivlist}}

\newenvironment{remark}[1][Remark]{\begin{trivlist}
\item[\hskip \labelsep {\bfseries #1}]}{\end{trivlist}}
\usepackage{setspace}

\ifCLASSINFOpdf
\else
\fi
\begin{document}
\title{Control Capacity of Partially Observable\\ Dynamic Systems in Continuous Time} 
\author{Stas~Tiomkin, Daniel~Polani, Naftali~Tishby
\thanks{ The Rachel and Selim Benin School of Computer Science and Engineering,The Hebrew University, Givat Ram 91904, Jerusalem, Israel, stas.tiomkin@mail.huji.ac.il}
\thanks{School of Computer Science, University of Hertfordshire, Hatfield AL10 9AB, United Kingdom, d.polani@herts.ac.uk}
\thanks{The Edmond and Lilly Safra Center for Brain Sciences and The Rachel and Selim Benin School of Computer Science and Engineering, The Hebrew University, Givat Ram 91904, Jerusalem, Israel, tishby@cs.huji.ac.il}}%


\maketitle
\begin{abstract}
Stochastic dynamic control systems relate in a probabilistic fashion the space of control signals to the space of corresponding future states. Consequently, stochastic dynamic systems can be interpreted as an information channel between the control space and the state space. In this work we study this {\it control-to-state} informartion capacity of stochastic dynamic systems in continuous-time, when the states are observed only partially.  
The control-to-state capacity, known as {\it empowerment}, was shown in the past \cite{Empowerment1,  Empowerment3, Empowerment4, Empowerment6,   Empowerment9,  Empowerment12} to be useful in solving various Artificial Intelligence \& Control benchmarks, and was used to replace problem-specific utilities. The higher the value of empowerment is, the more optional future states an agent may reach by using its controls inside a given time horizon. 

The contribution of this work is that we derive an efficient solution for computing the control-to-state information capacity for a linear, partially-observed Gaussian dynamic control system in continuous time, and discover new relationships between control-theoretic and information-theoretic properties of dynamic systems. Particularly, using the derived method, we demonstrate that the capacity between the control signal and the system output does not grow without limits with the length of the control signal. This means that only the {\it near-past window} of the control signal contributes effectively to the control-to-state capacity, while most of the information beyond this window is irrelevant for the future state of the dynamic system. We show that empowerment depends on a time constant of a dynamic system.
\end{abstract}
\IEEEpeerreviewmaketitle
\section{Introduction}
\IEEEPARstart{R}ecent advances in the understanding of complex dynamic systems reveal intimate connections between information theory and optimal control, \cite{InfoContr1, InfoContr2, InfoContr3, InfoContr4, InfoContr5, InfoDecisions, InfoDecisions2}.  To act optimally, engineering and biological systems must satisfy not only the requirements of optimal control, (e.g., getting close to a target state, tracking smoothly a nominal trajectory etc.), but increasingly other important constraints, such as the limitation of bandwidth, the restriction of memory or limited delays.
 
A stochastic dynamic control system is driven from state to state by appropriate control signals from the admissible set of control signals. All the system states achievable by some admissible control signal define the reachable set of a dynamic system. In this work we take an alternative view of reachability, by studying the relation between the admissible control and the reachable set of states in terms of information theory. For this, we consider stochastic dynamic systems as information channel, \cite{Cover}, between these two sets, or more precisely, between two random variables - the stochastic control, and the resulting state of dynamic system, \cite{RefControllability}.
Each of the reachable set states is achievable by some control signal from the admissible control state. As elaborated below, an important property of the correspondence between the reachable and the admissible sets is a number of different control signals which are distinguishable at a particular reachable set state. A larger value would indicate a more influential control set. This influence is cast in the language of information theory as a question of the capacity between the control signal and the future state of the dynamic system.

\subsection{Empowerment}
This control capacity, also known as empowerment,\cite{Empowerment1,  Empowerment3, Empowerment4, Empowerment6,   Empowerment9,  Empowerment12}, has been shown to worked well as a universal heuristic in many contexts by generalizing "controllability" in an information-theoretic sense. It provides a plausible framework for generating intrinsic (self-motivated) behaviour in  agents, eliminating the need for task-predefined, external rewards. 

As an intuitive example, consider the underpowered pendulum, \cite{Empowerment9}, driven by a {\bf low-energy} stochastic control process. It has a) stable and b) unstable equilibria positions, (the bottom  and the top position, respectively). In the bottom state, fewer states can be controllably reached than in the upright pendulum state. In the top state, while, on its own, unstable, with the help of a control signal, a richer set of separate "futures" (in a still precisely to define sense) can be controllably produced than in the bottom.  
The maximum number of distinguishable controls at the system output is given by the information capacity between the control random variable and the future state, viewed as a random variable depending on the (random) control signal. Here the dynamic system comprises the information channel.

{Among other uses, it constitutes an example for agent controllers implementing so-called "intrinsic motivation" principle \cite{Schmidhuber1, Oudeyer1, Ay1, Ay2}, a class of control algorithm which has recently received significant attention; controllers based on these principles substitute problem-specific utilities by generic measures only depending on system dynamics to produce "situation-relevant" controls.}

The suitability of the control-state capacity (empowerment) to implement an "intrinsic motivation"-style controller was explored in a series of theoretical and practical studies \cite{Empowerment1,  Empowerment3, Empowerment4, Empowerment6,   Empowerment9}. 

The empowerment landscape is computed for a dynamical system and an agent is driven along this landscape as to maximize the local empowerment gradient. Typical effects include (but are not limited to) the agent being led to unstable equilibria and being stabilized there. The control-state capacity produces pseudo-utility landscapes from the dynamics of the system and is a is a promising approach in designing a priori (pre-task) controls for artificial agents.

Until now, little work has been done in the domain of continuous space and none in fully continuous (rather than discretized) time. Here, actions need to be constrained by power, and, up to the present paper, no relation between favoured time horizon and control signal power in the continuous time was known. Here we contribute by 1) calculating empowerment for continuous time and space (continuous time was not considered in the past), 2) computing the full solution for the linear control case, 3) link time horizon characteristics and power. 

To evaluate the empowerment of a partially observable dynamic system means to compute the channel capacity between the control and the final state, i.e. the maximum of the mutual information. We propose to compute the maximum of the mutual information between the control process trajectory, which we denote by $u(0\rightarrow T)$ or briefly by $u(\cdot)$, and the resulting system output, $y(T)$, where the system dynamics is perturbed by the process noise $\eta$, and the system output is perturbed by the sensor noise $\nu$.
\tikzstyle{int}=[draw, fill=none, minimum size=4em]
\tikzstyle{init} = [pin edge={to-,thin,black}]
\begin{figure}[h]
\begin{center}
\begin{tikzpicture}[node distance=2.5cm,auto,>=latex']
    \node [int] (a) {$\begin{aligned}\dot{x}(t)=&f(x(t), u(t), \eta(t))\\y(t)=&g(x(t), u(t), \nu(t))\end{aligned}$};
    \node (b) [left of=a,node distance=3.5cm, coordinate] {a};
    \node (end) [right of=a, node distance=3.5cm, coordinate]{b};
    \path[->] (b) edge node {$u(0\rightarrow T)$} (a);
    \path[->] (a) edge node {$y(T)$} (end);
\end{tikzpicture}
\caption[Figure]{The dynamic system, $f(\cdot)$, is driven by the uncontrolled noise $\eta(t)$, and the stochastic control $u(0\rightarrow T)$ to the future state $x(T)$, which is partially observed by $y(T)$ through the sensor $g(\cdot)$ within the observation noise $\nu(t)$.}
\label{channelFig}
\end{center}
\end{figure}
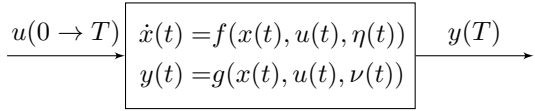
The relation between the control trajectory, $u(0\rightarrow T)$ and the future observed state, $y(T)$ is is thus stochastic and $u(\cdot)$ and $y(T)$ are jointly distributed according to Fig. \ref{InfochannelFig}.
\begin{figure}[h!]
\begin{center}
\begin{tikzpicture}[node distance=2.5cm,auto,>=latex']
    \node [int] (a) {$p\left(y(T)\mid u\left(0\rightarrow T\right)\right)$};
    \node (b) [left of=a,node distance=3.5cm, coordinate] {a};
    \node (end) [right of=a, node distance=3.5cm, coordinate]{b};
    \path[->] (b) edge node {$u(0\rightarrow T)$} (a);
    \path[->] (a) edge node {$y(T)$} (end);
\end{tikzpicture}
\caption[Figure]{An information channel, given by the conditional probability distribution, $p$, which is induced by the dynamic system, $f$, and the sensor, $g$, in Fig. \ref{channelFig}.}
\label{InfochannelFig}
\end{center}
\end{figure}
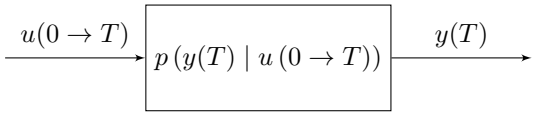
Empowerment of a stochastic dynamical system is given by 
\begin{align}
\mathcal{C}^* = \underset{p(u(\cdot))}{\max}I\left[u(\cdot);y(T)\right]
\end{align}
Apart from extending the applicability of the formalism to continuous time and arbitrary linear systems, the present analysis contributes new insights to the multifaceted marriage between information theory and optimal control.

Empowerment, in its original definition, requires the specification of its time horizon. This is the essential free parameter of the empowerment formalism. One question which regularly arises is how to choose this horizon. Our study, amongst other, shows that a characteristic time horizon may emerge through the dynamics only.

Communication-limited control is another important concept in modern engineering where the control-state capacity is useful is communication-limited control. The control-state capacity considers the channel between the controller to the future state, while most of the work dealing with communication-limited control, e.g. \cite{InfoContr1, InfoContr2}, is focused on the channel between the state sensor to the controller.  

For the first channel, the central question is: {\it what is the minimal information from the state to the controller that is required in order to satisfy a control objective?}. For the second channel, the central question is its dual: {\it what is the maximal information that can be transferred from the controller to the state?"}. 

These questions are complementary questions within the action-perception cycle paradigm, addressing the sensing and the actuation aspects of the cycle, respectively, \cite{InfoDecisions}. Obviously, a comprehensive study of the control under communication constraints needs to consider also the dual, controller-to-state channel.

These are only a selection of a multitude of directions where the control-state capacity has direct implications and provides new understandings. In this paper we concentrate on studying the intrinsic properties of this channel, relating the information-theoretic properties with properties of optimal control. In particular, we derive the relation between the energy, the time-horizon, and the capacity of the control-state of the dynamic system. 

The paper is organized as following. In Section \ref{Problem Definition} we define the problem of the information capacity between the control signal and the future output. The solution is based on the bi-orthogonal representation of Gaussian process. We establish the optimality conditions, KKT-conditions, for the optimal variances of the Gaussian process expansion, and we show that the optimal set of the variances can be found efficiently by an iterative water-filling algorithm. The second component of the solution is formed by the expansion functions of the Gaussian process expansion. We derive an optimality condition for the expansion functions. This optimality condition is satisfied by the spectral decomposition of the controllability Gramian of the linear control system. In the rest of the paper we elaborate on the details of the solution. 

Particularly, In Section \ref{GauProSec} we review briefly the bi-orthogonal expansion of Gaussian processes. In Section \ref{sec:MutInfo} we define the mutual information for continuous-time linear control systems. In Section \ref{sec:OptProb} we define the optimization problem, which solves the problem defined in Section \ref{Problem Definition}. In Section \ref{sec:OptEq} we elaborate on the optimality conditions of the optimization problem, where we derive the KKT-optimality conditions and set the conditions for the optimality of the Gaussian process expansion functions. In Section \ref{sec:AsymAna} we provide the asymptotic analysis of the control-to-output capacity. In Section \ref{sec:Simul} we provide computer simulations of the developed method and explain their results. And, finally, in Section \ref{sec:Summ} we summarise the work.

\section*{Main Result}
We explore the dependency of the control capacity on the time horizon, $T$, and the power, $P$. We show how to compute efficiently the information capacity, $C(P, T)$, between the control signal and the future output in continuous-time linear control systems. The solution is based on the bi-orthogonal expansion of the Gaussian process. This enables us to decompose the solution procedure into two independent aspects: the 'iterative water-filling' algorithm, and the spectral decomposition of the controllability Gramian. 

In our example, we find particularly that  the control input signal contributes to the empowerment value only during a temporally limited phase. Beyond this limited time window, the future output  is not essentially influenced by the control signal. This finding is consistent with intuition, as memoryless linear control systems 'forget' far past input events. 

We establish a quantitative relation between empowerment and the intrinsic features of dynamic systems, such as the time constant, $\tau$. Particularly, we show the control capacity achieves a finite limit for $\tau \rightarrow 0$. While this is, on first sight, we explain this effect mathematically, and provide a physical intuition.    

We argue that various results of this work will permit the study of nonlinear control systems as well by locally approximating them by linear systems. 

\section*{Notation}\label{notationSec}
The following notation is used in the paper. Uppercase and lowercase letters denote real matrices and vectors, respectively. $X^{'}$ denotes the (real-values) transpose of $X$. $x_{[: i]}$ is the $i$-th column of matrix $X$, $\mathbf{Tr}\left\{\cdot\right\}$ denote the matrix trace and $\mathbf{E}\left[\cdot\right]$ denotes the expectation operator. The vertical bars,$\left|\cdot\right|$, denote the matrix determinant. The notation $x_{\widehat{i}}$, and $x_{i}$ mean all the components of $x$ excluding the $i$-th component, and the $i$-th component of $x$, respectively. 

We always assume the initial state of the dynamical system to be given and known, and thus, by abuse of notation, we will never write explicitly the conditioning with respect to this initial state.

\section{Problem Definition}\label{Problem Definition}
Consider the following constrained optimization problem:
\begin{align}
&\underset{p(u(\cdot))\in\mathscr{P}}{\max}I\left[u(\cdot);y(T)\right]\\
&\mbox{s.t. }\;\;\dot{x}(t) = Ax(t) + Bu(t) + G\eta(t),\\
&\qquad {y}(t) = Cx(t) + F\nu(t),
\end{align}
where $I\left[\cdot;\cdot\right]$ is the mutual information between the $n$-dimensional vector $x_{n\times 1}(T)$ at a given time $T$, and the $p$-dimensional Gaussian control process\footnote{in this paper we consider w.l.og. zero-mean Gaussian processes, because the mean does not affect the mutual information.}, $u_{p\times 1}(\cdot)$, over times $t\in[0,T]$, { distributed according to the probability density function, $p(u(\cdot))$, which is restricted to the space of Gaussian process distribution functions, $\mathscr{P}$}. 
The dependency of $x(T)$ on $u(\cdot)$ is restricted by the dynamic affine-control system with the dynamics matrix, $A_{n\times n}$, the control process gain matrix, $B_{n\times p}$, and the process noise gain matrix $G_{n\times n}$, scaling the white Gaussian process noise, $\eta(t)$, which has a given autocorrelation function: 
\begin{align}\label{noiseAutoCor}
R_{\eta}(t_1, t_2)= \sigma_{\eta}\delta(t_1, t_2).  
\end{align}
where $\delta(t_1, t_2)$ is the delta function, and $\sigma_{\eta}$ is the noise power. The sensor noise, $\nu(t)$, is assumed to be the white Gaussian noise with the autocorrelation function:
\begin{align}
R_{\nu}(t_1, t_2)= \sigma_{\nu}\delta(t_1, t_2). 
\end{align}
 
The optimization takes place over the space of the Gaussian process probability distributions, $\mathscr{P}$, or, equivalently, over the space of the autocorrelation functions of the control process, $R_{u}(t_1, t_2)=\mathbf{E}\left[u(t_1)u^{'}(t_2)\right]$, because a Gaussian process is defined uniquely by its autocorrelation function. We assume the process noise, $\eta(t)$, to be  independent of the control signal, $u(t)$. We specifically do not restrict the control process to be a stationary process.

The mutual information between continuous variables would be unbounded unless their power is limited. To render the question well-defined, we therefore look for the maximum of the mutual information, {\it the capacity}, under the constraint of a given maximum total\footnote{a pointwise maximal power might also be considered.} power of the control process: 
\begin{align}
&\underset{R_u(\cdot, \cdot)}{\mbox{{\it max }}}I\left[u(\cdot);y(T)\right]\label{optProblem}\\
&\mbox{s.t. }\;\dot{x}(t) = Ax(t) + Bu(t) + G\eta(t),\label{optProblemConstr}\\
&\qquad y(t) = Cx(t) + F\nu(t),\label{optProblemConstr2}\\
&\mbox{     }\quad\;\;\mathbf{Tr}\left\{\int_0^TR_u(t, t)dt\right\}\le P.
\end{align}
The optimization over the space of the autocorrelation functions corresponds to  an optimization over the space of symmetric-positive definite functions, 
\begin{align}
\forall t_1, t_2\in[0,T] :& R_u(t_1, t_2) = R_u(t_2, t_1),\label{R_symmetry}\\
\forall h\in L^2([0,T]) :& \int_0^T\!\!\!\!\!\int_0^T\!\!\!\!\!R_u(t_1, t_2)h(t_1)h(t_2)dt_1dt_2\ge 0\label{R_psd},
\end{align}
which introduce an uncountable set of constraints to the problem in (\ref{optProblem}). We propose to perform the optimisation in (\ref{optProblem}) with regards to the components of the bi-orthonormal expansion of Gaussian process instead of $R_u(\cdot, \cdot)$, as described in the next section. 

The capacity in (\ref{optProblem}), $C$, depends on the control power constraint, $P$, the time horizon, $T$, and the dynamic system matrices, $A$, $B$, $G$, and $F$. The goal of this work is to compute this capacity efficiently, and to explore its properties with regard to the free parameters of the problem, and the ensuing properties of the dynamical control system.

\section{Control Process Representation}\label{GauProSec}
In this next section we represent the control process using the bi-orthogonal expansion of Gaussian processes, \cite{Gallager}, which helps to recast the constraints in (\ref{R_symmetry}) and (\ref{R_psd}) in a more convenient form for the optimization. 

It can be shown that any zero-mean Gaussian process $u(t)$ may be represented by an appropriate choice of $\left\{g_i(t)\right\}_{i=1}^{\infty}$, where $\left\{g_i(t)\right\}_{i=1}^{\infty}$ is a countable set of real orthonormal functions and $u_i$ is a sequence of independent Gaussian random variables.

Then, any zero-mean\footnote{a mean can be added if desired} Gaussian process, $u(t)$, may be represented by an appropriate choice of $\left\{g_i(t)\right\}_{i=1}^{\infty}$, and $\left\{u_i\right\}_{i=1}^{\infty}$ as following:
\begin{align}
u(t) = \sum_{i=1}^{\infty} u_ig_i(t).\label{controlBiExt}
\end{align}

This representation is known as the {\it bi-orthogonal expansion} of Gaussian process. It is bi-orthogonal because the expansion functions, $\left\{g_i(t)\right\}_{i=1}^{\infty}$, are orthogonal, and the random variables $\left\{u_i\right\}_{i=1}^{\infty}$ are statistically independent. Using the bi-orthogonal representation, we can represent the $p$-dimensional control process as following. 
\begin{align}
\vec{u}(t) = \left[\begin{matrix}\sum_iu_{i1}g_{i1}(t)\\ \vdots \\ \sum_iu_{ip}g_{ip}(t)\end{matrix}\right]_{p\times 1}.\label{controlVec}
\end{align}
We assume that the components of the control process vector, $\vec{u}(t)$, are independent. Consequently, the autocorrelation function of $\vec{u}(t)$ is diagonal, whose $kk$-th entry is specified by the parameters $\sigma_{ik}$ of the control process as follows:
\begin{align}
\left[R_u(t_1, t_2)\right]_{kk} = \sum_{i}\sigma_{ik}g_{ik}(t_1)g_{ik}(t_2).\label{controlProcessAutoRepr}
\end{align}
Consequently, the power constraint in (\ref{optProblemConstr}) appears as:
\begin{align}
\mathbf{Tr}\left\{\int_0^TR_u(t, t)dt\right\} = \sum_{m=1}^p\sum_{i=1}\sigma_{im}\le P,\label{totPowerConstr}
\end{align}
and, the following orthogonality constraints on the function set $\left\{g_{im}(t)\right\}_{im}$ are:
\begin{align}
\forall m,i,j  \enspace : \enspace \int_0^T g_{im}(t)g_{jm}(t)dt = \delta_{ij}\label{funOrthoConstr}.
\end{align}
\section{Mutual Information}\label{sec:MutInfo}
The mutual information between continuous variables is defined, e.g.\cite{Cover}, by the difference between the differential entropy of one variable, $h(\cdot)$ and the conditional differential entropy of this variable with respect to the other, $h(\cdot\mid\cdot)$. For Gaussian distributions, these are functions of the covariances of the corresponding random variables. In our case, we are interested specifically in the mutual information between the control signal over a time period $T$ and the resulting sensor signal at the end of that period:
{\small
\begin{align}
I\Bigl[u(0\rightarrow  T); y(T)\Bigr] =& h\Bigl[y(T)\Bigr] - h\Bigl[y(T)\mid u(0 \rightarrow T)\Bigr]\nonumber\\
=&\ln\Bigl(\Bigl|\Sigma_{y}(T)\Bigr|\Bigr)  - \ln\Bigl(\Bigl|\Sigma_{y\mid u(\cdot)}(T)\Bigr|\Bigr)\label{mutInfo}\\
=&\ln\Bigl(\Bigl|\frac{\Sigma_{y}(T)}{\Sigma_{y\mid u(\cdot)}(T)}\Bigr|\Bigr).
\end{align}
}
where $\Sigma_{y}(T)$, and $\Sigma_{y\mid u(\cdot)}(T)$ are directly found, \cite{LinSys}, from the solution\footnote{we assume w.l.o.g zero initial condition, which can be absorbed into the process noise covariance.} to the linear dynamic system equations in (\ref{optProblemConstr}) and (\ref{optProblemConstr2}), given by:
{\small
\begin{align}
y(T) = C\!\!\!\int_0^T\!\!\!\!\!e^{A(T-t)}Bu(t)dt + C\!\!\!\int_0^T\!\!\!\!\!e^{A(T-t)}G\eta(t)dt +F\nu(T)dt
\end{align}
}
Particularly, under the noise independence assumption:    
\begin{align}
&\Sigma_{x}(T) = \Sigma_{u}(T) + \Sigma_{\eta}(T),\label{VarX}\\
&\Sigma_{y}(T) = C\Sigma_{x}(T)C^{'} + \Sigma_{\nu}(T),\label{VarY1}
\end{align}
where
\begin{align}
&\Sigma_{u}(T) = \int\limits_0^T\int\limits_0^Te^{A(T-\tau_1)}BR_{u}(\tau_1, \tau_2)B^{'}e^{A^{'}(T-\tau_2)}d\tau_1d\tau_2,\nonumber\\
&\Sigma_{\eta}(T) ={\int\limits_0^T\int\limits_0^Te^{A(T-\tau_1)}GR_{\eta}(\tau_1, \tau_2)G^{'}e^{A^{'}(T-\tau_2)}d\tau_1d\tau_2},\nonumber\\
&\Sigma_{\nu}(T) = T\sigma_{\nu},\label{sensorCov}
\end{align}
where the last term is due to the sensor noise, which is not convolved with the system dynamics, but rather expands for $T$ time units as an unconstrained Wiener process with the variance $\sigma_{\nu}$. 
We will denote by $\Sigma_n$ the total covariance of uncontrolled noise. Under the noise independence assumption, the total uncontrolled covariance is the sum of the process noise and the sensor noise:
\begin{align}
\Sigma_n(T) \doteq C\Sigma_{\eta}(T)C^{'} + \Sigma_{\nu}(T),\label{totNoise}
\end{align}
which is also the covariance of $y(T)$, conditioned on the control process, $u(\cdot)$:
\begin{align}
\Sigma_{y\mid u(\cdot)}(T)  = \Sigma_n(T),
\end{align}
which holds, because, knowing the control process, the only uncertainty in $y(T)$ is due to the noise. 
Following the uncontrolled noise covariance definition in (\ref{totNoise}), the future observable state covariance matrix in (\ref{VarY1}) is:
\begin{align}
\Sigma_{y}(T) = C \Sigma_{u}(T) C^{'} + \Sigma_n(T) \label{VarY}.
\end{align}
Using the definitions of the control and the noise autocorrelation functions given by (\ref{noiseAutoCor}) and (\ref{controlProcessAutoRepr}), respectively, we have: 
{\small
\begin{align}
\Sigma_u(T) =  \sum_{m=1}^p\sum_i\sigma_{{im}}&\int_0^T\Bigl(e^{A(T-t_1)}b_mg_{{im}}(t_1)dt_1\Bigr)\nonumber\\
\cdot&\Bigl(\int_0^T e^{A(T-t_2)}b_mg_{im}(t_2)dt_2\Bigr)^{'},\label{ControlGramian}\\
\Sigma_{\eta}(T) = \sigma_{\eta}\int_0^Te^{A(T-t)}&GG^{'}e^{A^{'}(T-t)}dt \succ 0 .\label{NoiseGramian}
\end{align}
}
where $b_m = B_{[:, m]}$.  
The following notations will be used in the paper:
\begin{align}
&w_{A,b_m}(t) \doteq e^{A(T-t)}b_m,\\
&z_{im}(T)\doteq \int_0^Tw_{A,b_m}(t)g_{im}(t)dt,\label{Zvectors}\\
&W(A, b_m, T) \doteq \int_0^Tw_{A,b_m}(t)w^{'}_{A,b_m}(t)dt\succ 0,\label{GramianNotation}
\end{align} 
having dimensions $n\times 1$, $n\times 1$, and $n\times n$, respectively. Equation (\ref{GramianNotation}) is the controllability Gramian of the dynamic system. We assume the system is controllable, which means $W(A, b_m, T)$ is of full rank, \cite{LinSys}. 

Following the full representation of the control process in (\ref{controlBiExt}), we introduce a partial control signal representation without the $im$-th control component:
\begin{align}
u_{\widehat{im}}(\cdot) \doteq& \sum\limits_{j,k\; \neq \;i,m }u_{im}(\cdot),\label{controlBiExtPart}
\end{align}
where, by abuse of notation, we write briefly,
\begin{align}
u_{im}(\cdot) \doteq& u_{im}g_{im}(\cdot),\label{controlBiExtPart2}
\end{align}
where $u_{im}$ is as in (\ref{controlVec}). 
To improve the readability we will omit sometimes the time dependence, denoting $z_{im}(T)$ by $z_{im}$. Using the above notations, the control process covariance matrix in (\ref{ControlGramian}) appears as
{\small
\begin{align}
\Sigma_{u}(T) = \sum_{m=1}^{p}\sum_i\sigma_{im}{z}_{im}(T){z}^{'}_{im}(T)\label{SigmaUT}.
\end{align}
}
{
Consequently, the final state covariance matrix appears as:
{\small
\begin{align}
\Sigma_{x}(T) =  \sum_{m=1}^{p}\sum_i\sigma_{im}{z}_{im}(T){z}^{'}_{im}(T) +  \Sigma_{\eta}(T) \label{Q}.
\end{align}
}
The control process in (\ref{controlVec}) is represented by the collection of the independent components, $\{u_{im}\}_{im}$. While a particular component is denoted $u_{im}$. We will adopt the notation $u_{\widehat{im}}$ for the collection of component without the particular component $u_{im}$. It will be seen useful to consider the partial control process covariance matrix:
{\small
\begin{align}
\Sigma_{u_{\widehat{im}}}(T) = \Sigma_{u}(T) - \sigma_{im}{z}_{im}(T){z}^{'}_{im}(T)\label{SigmaUT_part}.
\end{align}
}
Consequently, the final state covariance, conditioned on the $im$-th control process component, is: 
{\small
\begin{align}
\Sigma_{x\mid u_{{im}}}(T) = \Sigma_{u_{\widehat{im}}}(T) + \Sigma_{\eta}(T)\label{Q_miDef},
\end{align}
}
which will be denoted in the following by $\Sigma_{x\mid u_{{im}}}$. The equation (\ref{Q_miDef}) follows from the linearity and the independence between different control process components. Intuitively, knowing a particular control component $u_{{im}}$ reduces the overall uncertainty in $x(T)$ to the uncertainty due to $u_{\widehat{im}}$ and the noise. The final observable state covariance, conditioned on the $im$-th control process component, is given by:
{\small
\begin{align}
\Sigma_{y\mid u_{{im}}}(T) = C  \Sigma_{u_{\widehat{im}}}(T)  C^{'}  + \Sigma_{n}(T)\label{Q_miDef},
\end{align}
}
} 
{The mutual information given in (\ref{mutInfo}) would be unbounded without further assumptions on $\left\{\sigma_{im}\right\}_{im}$ and $\left\{g_{im}(t)\right\}_{im}$. The assumption we use here is that the total sum of $\left\{\sigma_{im}\right\}_{im}$ is bounded by a positive constant $P$, and that the $\left\{g_{im}(t)\right\}_{im}$ are continuous over the closed interval $[0,T]$. With these constraints given, we proceed to formulate the optimization problem in the next section.}

\section{Optimization Problem}\label{sec:OptProb}
{Applying the particular expressions for the mutual information in linear dynamic systems, (\ref{mutInfo}), and the expressions for energy and orthogonality constraints of the control process constraints, (\ref{totPowerConstr}) and (\ref{funOrthoConstr}), the optimization problem (\ref{optProblem}) reduces to the following\footnote{we omit the constant factors that doesn't  affect the maximization.}:
{\small
\begin{align}
&\underset{\vec{\sigma}, \left\{g_{im}(\cdot)\right\}_{im}}{\mbox{{\it max  }}}\ln\left(\frac{\left|\Sigma_{n}(T) + \sum_{m=1}^p\sum_i\sigma_{im} Cz_{im}z^{'}_{im}C^{'}\right|}{\left|\Sigma_{n}(T)\right|}\right)\label{MIObjective}\\
&\quad{\mbox{{\it subject to}}}\enspace\sum_{m=1}^p\sum_i\sigma_{im} = P,\label{SigmaConstr1}\\
&\quad\qquad\qquad\enspace\forall i,m :\enspace \sigma_{im}\ge 0,\label{SigmaConstr2}\\
&\quad\qquad\qquad\enspace\forall i,m :\enspace \int_0^T g_{im}(t)g_{jm}(t)dt = \delta_{ij}\label{gConsrt},
\end{align}
}
}
where $\vec{\sigma}$ is a shortcut notation for  $\vec{\sigma}\doteq\{\sigma_{im}\}_{im}$. The corresponding unconstrained Lagrangian, $\mathcal{L}$, is:
\begin{align}\label{LagrangianDef}
\mathcal{L} =& \mathcal{L}\left[\vec{\sigma}, \{g_{im}(\cdot)\}_{im}, \lambda, \{\nu_{mij}\}_{mij}, \{\gamma_{im}\}_{im}\right],
\end{align}
\noindent{where $\lambda$, and $\nu_{mij}$, $\gamma_{im}$ are the Lagrange multipliers for the power, (\ref{totPowerConstr}), the orthogonality constraints, (\ref{funOrthoConstr}), and the variance positivity constraints, respectively. The Lagrangian becomes:} 
{\small
\begin{align}
\mathcal{L}=&\ln\left(\frac{\left|\Sigma_{n}(T) + \sum_{m=1}^p\sum_i\sigma_{im} Cz_{im}z^{'}_{im}C^{'}\right|}{ \left|\Sigma_{n}(T)\right|  }\right)\nonumber\\
&- \lambda\left(\sum_{m=1}^p\sum_i\sigma_{im} - P\right)
- \sum_{m=1}^p\sum_i\gamma_{im}\sigma_{im}\nonumber\\
&- \sum_{m=1}^p\sum_{ij}\nu_{m,i,j}\left( \int_0^T g_{im}(t)g_{jm}(t)dt - \delta_{ij}\right)\label{Lagrangian} 
\end{align}
}
The corresponding KKT optimality conditions are:
\begin{align}
\forall m,i,t\enspace:\enspace& \frac{\delta \mathcal{L}[g_{im}]}{\delta g_{im}}(t) = 0,\label{varDirG}\\
\forall m,i\enspace:\enspace& \frac{\partial\mathcal{L}}{\partial\sigma_{im}} = 0,\label{sigmaDer}\\
\forall m,i,j\enspace:\enspace& \int_0^T g_{im}(t)g_{jm}(t)dt = \delta_{ij},\label{nuDer}\\
\forall m,i\enspace:\enspace& \gamma_{im}\sigma_{im} = 0,\enspace\gamma_{im}\ge0,\enspace\sigma_{im}\ge0,\label{KKT}\\
&\sum_{m=1}^p\sum_i\sigma_{im} = P\label{lambdaDer}.
\end{align}
where (\ref{KKT}) is the complementarity KKT-condition. 
The optimization problem in (\ref{MIObjective}-\ref{gConsrt}) is a convex optimization problem with regard to $\{\sigma_{im}\}_{im}$ for a given set of the expansion functions, $\{g_{im}(t)\}_{im}$, which can be solved numerically by the existing convex optimization solvers, \cite{Boyd}. For example, it can be found iteratively by the coordinate ascent algorithm, starting from an arbitrary parameters set, $\{\sigma_{im}\}$, and climbing each time along different direction until convergence to the global maximum, which is guaranteed for any starting point, because the objective function in (\ref{MIObjective}) is concave. However, the properties of the mutual information enable to derive a formal solution for $\{\sigma_{im}\}_{im}$, as explained below.

The objective function, the mutual information in (\ref{MIObjective}), can be decomposed into individual components, each of which expresses the mutual information between a particular Gaussian control process component and the future state. We prove this in the next section. This separation provides an insight to the nature of the optimal control signal, and will enables us to derive an implicit solution for $\{\sigma_{im}\}_{im}$. 
 
\begin{remark}
n Section.\ref{Gramian Decomposition} we show that, despite having introduced the expansion of $g$ over $i$ as infinite, only a finite number of $z$-vectors (namely $n$) are in fact required to satisfy the optimality condition of the Lagrangian (\ref{LagrangianDef}) with regard to $\{g_{im}(\cdot)\}_{im}$. Therefore, in Eq. (\ref{SigmaUT}), we can limit ourselves to consider just the summation over $i=1,..,n$. 
\end{remark}
\subsection{Information Decomposition}
The mutual information in (\ref{MIObjective}) can be decomposed to a sum of mutual information terms between different control signal components: 
{\small
\begin{align}
I[u(\cdot);y(T)]= \frac{1}{np}\sum_{m=1}^p\sum_{i=1}^n\Bigl\{ &I[u_{{im}}(\cdot);y(T)]+\\
&+ I[u_{\widehat{im}}(\cdot);y(T)\mid u_{{im}}(\cdot)]\Bigr\}\label{MU_ExclIncl},
\end{align}
}
where $I[u_{{im}}(\cdot);y(T)]$, and $I[u_{\widehat{im}}(\cdot);y(T)\mid u_{{im}}(\cdot)]$ is the mutual information between the $im$-th control signal component and the future observable state, and the mutual information between the control signal without its $im$-th component and the future observable state, conditioned on the $im$-th control component, respectively.
The decomposition follows directly from the chain rule for the mutual information:
{\small
\begin{align}
I[U^{'}, U^{''};Y] = I[U^{'};Y] + I[U^{''};Y\mid U^{'}]. 
\end{align}
}
Consequently, the Lagrangian in (\ref{Lagrangian}) is equivalent to   
{\small
\begin{align}
\mathcal{L}= \frac{1}{np}&\sum_{m=1}^p\sum_{i=1}^n\Bigl\{ I[u_{{im}}(\cdot);y(T)]
+ I[u_{\widehat{im}}(\cdot);y(T)\mid u_{{im}}(\cdot)]\Bigr\}\nonumber\\
&- \lambda\left(\sum_{m=1}^p\sum_i\sigma_{im} - P\right)
- \sum_{m=1}^p\sum_i\gamma_{im}\sigma_{im}\nonumber\\
&- \sum_{m=1}^p\sum_{ij}\nu_{mij}\left( \int_0^T g_{im}(t)g_{jm}(t)dt - \delta_{ij}\right)\label{LagrangianEq}.
\end{align}
}
The advantage of this Lagrangian representation is due to the fact that the second term of the objective function, $I[u_{\widehat{im}}(\cdot);y(T)\mid u_{{im}}(\cdot)]$ does not depend on $\sigma_{im}$ by definition, while, in the first term, for all $i,m$ $\sigma_{im}$ is separated from all other coefficients, $\sigma_{j,k\;\neq \;i,m}$, which suggests an efficient way to compute the optimal parameters set, $\vec{\sigma}$.  
\section{Optimality Equations}\label{sec:OptEq}
\subsection{Expansion Variances - Generalized Water-filling.}
Computing the ordinary derivative of the Lagrangian in (\ref{LagrangianEq}) with regard to $\sigma_{im}$, and equating it to zero we get:
\begin{align}
\forall i,m\enspace:\enspace\frac{\partial \mathcal{L}(\sigma_{im})}{\partial\sigma_{im}} 
=& \frac{{z}^{'}_{im}C^{'}\Sigma_{y\mid u_{im}}^{-1}Cz_{im}}{1+\sigma_{im}{z}^{'}_{im}C^{'}\Sigma_{y\mid u_{im}}^{-1}Cz_{im}} - \lambda - \gamma_{im},\label{optimalityLagrSigma}
\end{align}
{
{Due to the complementary slackness constraint (\ref{KKT}), we have:}
\begin{align}
\sigma_{im}>0\Rightarrow& \frac{{z}^{'}_{im}C^{'}\Sigma_{y\mid u_{im}}^{-1}Cz_{im}}{1+\sigma_{im}{z}^{'}_{im}C^{'}\Sigma_{y\mid u_{im}}^{-1}Cz_{im}} - \lambda = 0 ,\label{lambda_mOptCondGrt0}\\
\sigma_{im}=0\Rightarrow& {{z}^{'}_{im}C^{'}\Sigma_{y\mid u_{im}}^{-1}Cz_{im}}- \lambda  -\gamma_{im}  = 0 ,\label{lambda_mOptCond}
\end{align}
Consequently, 
{\small
\begin{align}
\sigma_{im} = \max\left(0, \frac{1}{\lambda} -\frac{1}{z^{'}_{im}C^{'} \Sigma_{y\mid u_{im}}^{-1}{(\vec{\sigma})}Cz_{im}}\right),\label{fixedpointupdate} \\
\intertext{{\normalsize where $\Sigma_{y\mid u_{im}}{(\vec{\sigma})}$ according to (\ref{Q_miDef}) is, (we here explicitly write $\vec{\sigma}$ to make the dependence on the $\sigma$ parameters explicit, and drop writing the dependency on $T$)}}
\Sigma_{y\mid u_{im}}{(\vec{\sigma})} = \sum\limits_{jk\neq im}\sigma_{jk} Cz_{jk}z_{jk}^{'}C^{'} + \Sigma_{n}.
\end{align}
}
Altogether with the power constraint (\ref{lambdaDer}) we get:
{\small
\begin{align}\label{waterfilling} 
\sum_{im} \max\left(0, \frac{1}{\lambda} -\frac{1}{z^{'}_{im} C^{'}\Sigma_{y\mid u_{im}}^{-1}(\vec{\sigma})Cz_{im}}\right)=P,
\end{align}
}
where the global 'water-line', $\lambda$, is adjusted by linear search. 
A particular control process variance, $\sigma_{im}$, given by (\ref{fixedpointupdate}), depends on all the other variances through $\Sigma_{y\mid u_{im}}^{-1}(\vec{\sigma})$. 
A solution to the equations above forms a unique optimum because of concavity. The formal solution given by (\ref{fixedpointupdate}) lends itself to implement an iterative water-filling scheme, \cite{BoydIterWF}, in the form of a two-phase fixed-point iteration. One phase updates the $\sigma_{im}$ (\ref{fixedpointupdate}), the other adapts the water-line (\ref{waterfilling}).
\begin{figure}[h!]
\centering 
\includegraphics[scale=0.35]{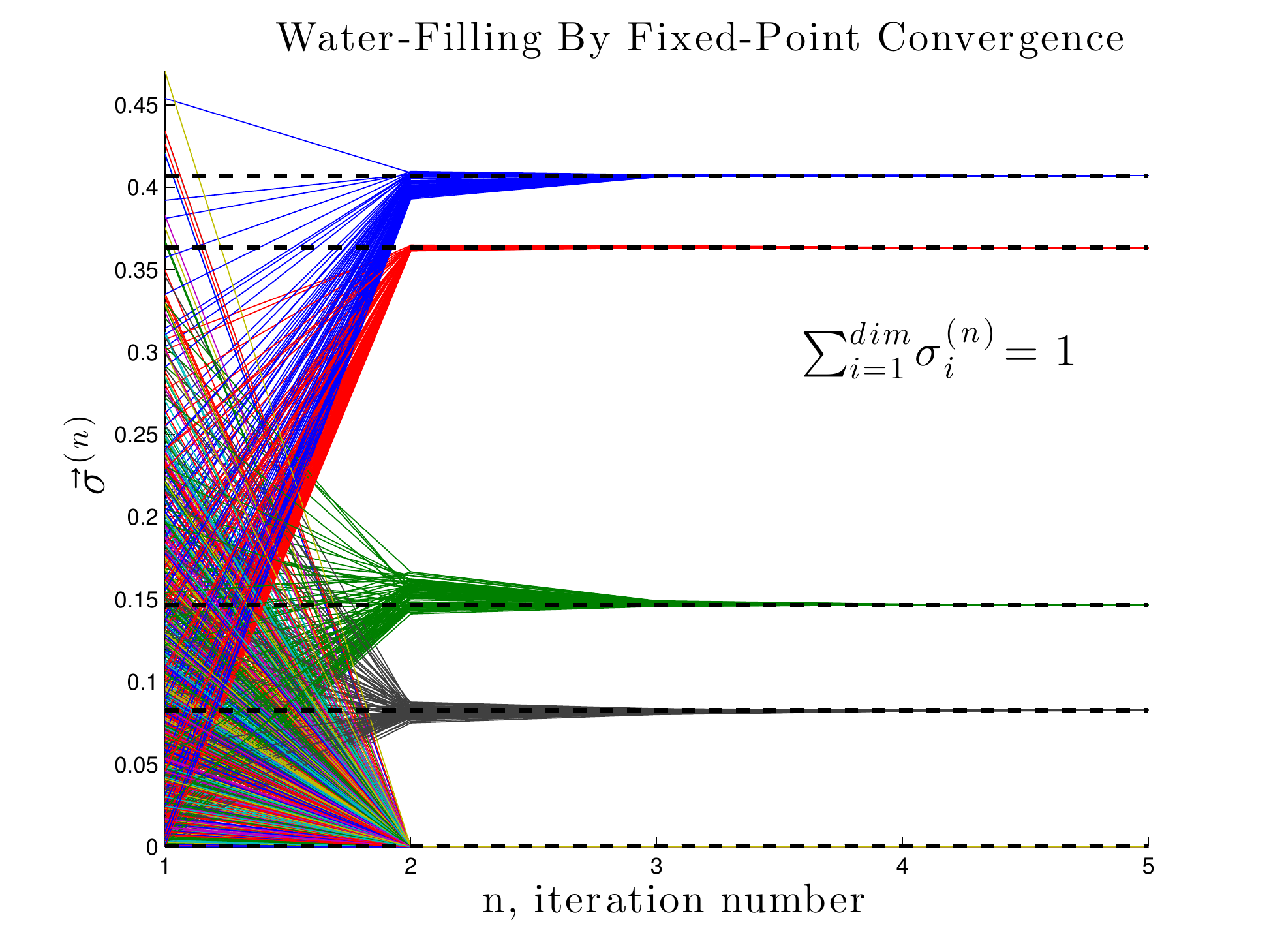} 
\caption{Demonstration of the convergence of $\vec{\sigma}^n$ to the leading components, where $n$ is an iteration number. The resulting non-zero components of the control process achieves the capacity for the given set of $z$-vectors. The total power is set to $P=1$.}
\label{SigmaCnvrg}
\end{figure} 
Fig.\ref{SigmaCnvrg} shows the typical convergence of this iteration, 'Iterative Water-Filling' scheme, for a randomly chosen dynamic system of dimension $n=10$, given by an arbitrary positive definite matrix $\Sigma_{n}\succ 0$, and an arbitrary  set of $z$-vectors. The y-axis shows the components of $\vec{\sigma}^{(n)}$, the x-axis the number of iterations, $n$. The algorithm was run 100 times for different initial conditions with the same random system matrix. The dashed-lines show the numeric solution found by an explicit convex optimization solver in comparison.
\subsection{Complementary condition for non-zero capacity}
Even though $\sigma_{im}$ is independent of $z_{im}$, it turns out that there is a complementary condition on $\sigma_{im}$ and $z_{im}$ which must be satisfied in order to get a non-zero capacity. This condition is stated and proved in the following lemma. 
\begin{lemma}\label{lambda_mOptCondLemma}
If $C$ has full rank, and $\exists(i, m)$ such that $\sigma_{im}>0$, but $z_{im}=0$, then $I\left[u(\cdot);y(T)\right] \equiv 0.$
\end{lemma}
\begin{proof}
Assume $\exists(i, m)$ such that $\sigma_{im}>0$ and $z_{im}=0$. $\Sigma_{y\mid u_{im}(\cdot)}^{-1}\succ 0$ implies that $z_{im}=0$ iff $z^{'}_{im}C^{'}\Sigma_{y\mid u_{im}(\cdot)}^{-1}Cz_{im}=0$. In this case (\ref{lambda_mOptCondGrt0}) implies that if $z^{'}_{im}C^{'}\Sigma_{y\mid u_{im}(\cdot)}^{-1}Cz_{im}=0$, then $\lambda=0$. Consequently, 
\begin{align}
\forall i,m : z^{'}_{im}C^{'}\Sigma_{y\mid u_{im}(\cdot)}^{-1}Cz_{im}=0\Rightarrow \forall i,m : z_{im} = 0,\label{lemmaKey}
\end{align}
which means 
\begin{align}
\underbrace{\sum_{i,m:\sigma_{im}>0}\sigma_{im}Cz_{im}{z^{'}}_{im}C^{'}}_{\mbox{=0 due to (\ref{lemmaKey})}} + \underbrace{\sum_{i,m:\sigma_{im}=0}\sigma_{im}Cz_{im}{z^{'}}_{im}C^{'}}_{\mbox{=0 due to }\sigma_{im}=0}=0.
\end{align}
Altogether it follows that, if $\exists(i, m):\sigma_{im}>0$ and $z_{im}=0$, then
\begin{align}
I\left[u(\cdot);y(T)\right]=0. 
\end{align}
\QEDA
\end{proof}
This lemma will be useful in the next section, where we derive the set of '$z$-vectors', $\left\{z_{im}\right\}_{im}$, which satisfies the KKT optimality conditions in (\ref{varDirG}-\ref{lambdaDer}). 

\subsection{Expansion Functions - 'Gramian Decomposition'.}\label{Gramian Decomposition}
In this section we derive the optimality conditions for the expansion functions, $\{g_{im}(t)\}$, by computing the functional derivative, \cite{VarDer}, of the Lagrangian (\ref{varDirG}). Particularly, with test functions $\phi(t)_{im}$: 
{\small
\begin{align}\label{FrechetDerivative}
\forall i,m\enspace:\enspace \int_0^T \frac{\delta \mathcal{L}[g_{im}]}{\delta g_{im}}(t)\phi_{im}(t)dt = \left[\frac{d\mathcal{L}[g_{im}+\epsilon\phi_{im}]}{d\epsilon}\right]_{\Bigr|_{\epsilon=0}}. 
\end{align}
}
As shown below, the optimality conditions for $\{g_{im}(t)\}$ enable to derive the corresponding 'z-vectors', $\left\{z_{im}\right\}_{im}$, defined in (\ref{Zvectors}).  The straightforward computation of the variation in (\ref{FrechetDerivative}) gives\footnote{the Jacobi's formula, ${\displaystyle {\frac {d}{dt}}\log \left|\Sigma_y(t)\right|=\mathbf{Tr} \left\{\Sigma_y(t)^{-1}\,{\frac {d}{dt}}\Sigma_y(t)\right\}}$, is used in the derivation.}:
{\small
\begin{align}
\forall m,i,t\enspace:\enspace \frac{\delta \mathcal{L}[g_{im}]}{\delta g_{im}}(t) = &
\sigma_{im} \mathbf{Tr}\left\{\Sigma_y^{-1} Cz_{im}  w^{'}_{A, b_m}(t)C^{'} \right\}\nonumber\\
+&\sigma_{im} \mathbf{Tr}\left\{\Sigma_y^{-1} Cw_{A, b_m}(t) z_{im}^{'}C^{'} \right\}\nonumber\\
+&2\sum_{j}\nu_{mij}g_{jm}(t)=0,
\end{align}
}
which, due to the symmetry of the trace, $\mathbf{Tr}(A)=\mathbf{Tr}(A^{'})$, and $\Sigma_y^{-1}$ reduces to
{\small
\begin{align}
\sigma_{im}\mathbf{Tr} \left\{\Sigma_y^{-1} Cz_{im}  w^{'}_{A, b_m}(t)C^{'}\right\}
=-\sum_{j}\nu_{mij}g_{jm}(t).\label{gipOpt}
\end{align}
}
$\forall m,i$ and $t\in[0,T]$.
As shown in Appendix \ref{App:AppendixB}, the equation (\ref{gipOpt}) is equivalent to the following equation
\begin{align}
\sigma_{im}\left(W(A, b_m, T) -\sum_{j} z_{jm} z^{'}_{jm}\right)C^{'}\Sigma_y^{-1} Cz_{im}=0,\label{zVecEq}
\end{align}  
which will be useful in order to derive the optimal set of 'z-vectors', $\left\{z_{im}\right\}$.
{It is worth to mention while the variation in (\ref{FrechetDerivative}) is computed with respect to the expansion functions, and $\left\{g_{im}(t)\right\}$, respectively. However, to compute the capacity we do not actually need to find the optimal expansion functions explicitly, but can express the capacity exclusively in terms of the set $\left\{z_{im}\right\}$.} 

The following lemma reveals the connection between the $z$-vectors and the controllability Gramian in (\ref{GramianNotation}). This result is important, because it enables to compute the capacity efficiently,  and it strengthens the interplay between the information-theoretic properties of the dynamic system and its optimal control properties.      
\begin{lemma}\label{theoremZvectors}
Choosing the $z_{jm}$ to the eigenvectors of the controllability Gramian satisfies equation (\ref{GramianNotation}). Particularly, set
\begin{align}
\forall i\in\{1,..,n\} : z_{im}(T) =& \sqrt{\omega_{im}(T)}v_{im}(T),\label{suffCond1}\\
\forall i>n : z_{im}(T) =& 0,\label{suffCond2}
\intertext{where $\omega_{im}(T)$ and $v_{im}(T)$ are the corresponding eigenvalues and normalized eigenvectors of the controllability Gramian, $W(A, b_m, T)\mbox{ given by (\ref{GramianNotation}).}\hfill\square$}\nonumber
\end{align}
\end{lemma}
\begin{proof}[Proof of lemma \ref{theoremZvectors}].
Equation (\ref{zVecEq}) is satisfied when at least one of the following cases hold: a) $\sigma_{im}=0$, b) $z_{im}=0$, c) $\tilde W=W(A, b_m, T)-\sum_jz_{jm}z^{'}_{jm}=0_{n\times n}$, d) $C^{'}\Sigma_y^{-1} Cz_{im}\in \mathrm{Null}(\tilde W)$.  W.l.o.g.\ there exists at least one $\sigma_{im}>0$, otherwise the capacity is zero. According to Lemma \ref{lambda_mOptCondLemma}, and assuming $C$ has full rank, the corresponding $z_{im}\neq 0$, and, consequently, $C^{'}\Sigma_y^{-1}Cz_{im}\neq 0$ too, because $\Sigma_y\succ 0$. In the case 'd)' the entire null space of $\tilde W$ is a solution to (\ref{zVecEq}). In the case 'c)' the solution is a decomposition of the Gramian into the sum of 1-rank matrices $z_{im}z_{im}^{'}$: 
\begin{align}
W(A, b_m, T)=\sum_jz_{jm}(T)z^{'}_{jm}(T).\label{relationOfZtoW}
\end{align}
The decomposition of a positive definite matrix to a sum of 1-rank matrices is not unique, and each of them solves (\ref{zVecEq}). A special case is the eigenvalue decomposition:
\begin{align}
W(A, b_m, T) =& V_{m}(T)\Lambda_{m}(T) V_{m}^{'}(T)\\
 =& \sum_{i=1}^n\omega_{im}(T)v_{im}(T)v^{'}_{im}(T)
\end{align}
with $\forall m:\omega_{im}(T)>0$, and $V(T)V^{'}(T)=I_{n\times n}$. All $\omega_{im}(T)>0$ is because the controllability Gramian is assumed to be of full rank. Choosing the $z_{jm}$ as eigenvectors of $W(A, b_m, T)$ completes the proof.
%
\QEDA
\end{proof}
\subsection{'Water Filling with Gramian Decomposition'}
Combining Lemma \ref{theoremZvectors} with (\ref{fixedpointupdate}), the water-filling expression for the expansion variances appears as, for $i\in\{1,..,n\}, m\in\{1,..,p\}$: 
{\small
\begin{align}
\sigma_{im} = \max\left(0,\frac{1}{\lambda} - \frac{1}{ \omega_{im}(T)v^{'}_{im}(T)C^{'}\Sigma_{y\mid u_{im}}^{-1}(\vec{\sigma})C v_{im}(T)  }\right),
\end{align}
}
where the 'water-line' is adjusted to satisfy the total power constraint,
\begin{align}
\sum_{i=1}^n\sum_{m=1}^p \sigma_{im} = P.
\end{align}
{ The capacity will in general depend on the time horizon, $T$. In the following section, we show that the capacity, $C(T)$, is limited by a finite value for $T\rightarrow \infty$:}
\begin{align}
\lim_{T\rightarrow \infty}C^{*}(T) = C^{\infty} < \infty.
\end{align}
And, the capacity is linear in time for $T\rightarrow 0$: 
\begin{align}
\lim_{T\rightarrow 0}C^{*}(T) \approx cT,
\end{align}
where $c$ is a positive constant which is derived from the properties of the dynamic system.
\section{Asymptotic Analysis}\label{sec:AsymAna}
\subsection{Infinite time horizon}
The total uncontrolled noise covariance matrix, $\Sigma_n(T)$, in (\ref{totNoise}) diverges for $T\rightarrow \infty$, due to the unlimited growth of the sensor noise, $\Sigma_{\nu}(T)$ (\ref{sensorCov}), with time, $T$, which means that the channel capacity would go to zero. More refined statements on the channel capacity for $T\rightarrow\infty$ in the open-loop case are possible, if one limits oneself to a perfect sensor, i.e. considering the special case $y(t)=x(t)$. We found that empowerment is finite both in stable systems, where all the eigenvalues of $A$ are negative, and in unstable systems, where all the eigenvalues of $A$ are positive.  
\subsubsection{Stable Systems}
In this section we provide the asymptotic analysis of the stable systems. When the system is stable, then the process noise variance matrix in (\ref{NoiseGramian}), (the noise Gramian), and the control controllability Gramian in (\ref{GramianNotation}) are finite for $T\rightarrow \infty$ and we can write:
\begin{align}
\Sigma_{\eta}(A, G, \infty) =& \sigma_{\eta}\int_0^{\infty} e^{A(T-t)}GG^{'}e^{A^{'}(T-t)}dt,\\
W(A, b_m,\infty) =& \int_0^{\infty}e^{A(T-t)}b_mb_m^{'}e^{A^{'}(T-t)}dt, 
\end{align}
and, they can be found analytically, (see for details e.g., Theorem 6.1 in \cite{LinSys}), by solving the corresponding continuous-time Lyapunov equations:
\begin{align}
A\Sigma_{n}(A, G, \infty) + \Sigma_{n}(A, G, \infty)A^{'} =& -GG^{'},\label{NoiseGramLyapunov}\\
\forall m:AW(A, b_m, \infty) + W(A, b_m, \infty)A^{'} =& -b_mb_m^{'}\label{ContrGramLyapunov},\\
W(A, b_m, \infty) = \sum_{i=1}^n\omega_{im}(\infty)v_{im}(\infty)&v^{'}_{im}(\infty).
\end{align}
In this case, the water-filling solution is given by:
{\small
\begin{align}
\sigma_{im}&(\infty) =\nonumber\\ 
&=\max\left(0,\frac{1}{\lambda} - \frac{1}{ \omega_{im}(\infty)v^{'}_{im}(\infty)\Sigma_{y\mid u_{im}}^{-1}(\infty) v_{im}(\infty)}\right),\label{WFS1}
\end{align}
}
This provides the asymptotic capacity, $C(\infty)$, which is finite for any finite power power constraint, $P$, where the noise variance matrix, $\Sigma_{n}(A, G, \infty)$, and the controllability matrix completely defines the capacity. 
The next section deals with the unstable systems, where we show that the asymptotic capacity is finite as well.
\subsubsection{Unstable Systems}
In this section we show that the asymptotic capacity is finite as well, if all the eigenvalues of $A$ are positive.
We need the following lemma in the asymptotic analysis of the capacity, $C(\infty)$, in systems where all the eigenvalues of $A$ are positive.
\begin{lemma}
The following noise and control Gramians,
\begin{align}
\tilde{\Sigma}_{\eta}(A, G, T) =& \sigma_{\eta}\int_0^{T} e^{-At}GG^{'}e^{-A^{'}t}dt,\label{GramianNoiseTilde}\\
\tilde{W}(A, b_m, T) =& \int_0^{T}e^{-At}b_mb_m^{'}e^{-A^{'}t}dt,\label{GramianControlTilde} 
\end{align}
define the same mutual information objective in (\ref{MIObjective}), as the original Gramians, $\Sigma_{n}(A, G, T)$ and $W(A, b_m, T)$ do.
\end{lemma}
\begin{proof}
For the optimal set of $\{z_{im}\}$ in (\ref{suffCond1}) and under the perfect sensor assumption above, the objective in (\ref{MIObjective}) appears as:
{\small
\begin{align}
\ln\Biggl(\Biggl|I_{n\times n} + \Sigma_{\eta}^{-1}(A, G, T) \sum_{m=1}^p W(A, b_m, T)\Biggr|\Biggr)\label{MIObjective2},
\end{align}
}
where $\Sigma_{\eta}(A, G, T)^{-1}$ and $W(A, b_m, T)$ can be represented by:
{\small
\begin{align}
\Sigma_{\eta}(A, G, T)^{-1} =& \left(\int_0^T e^{A(T-t)}GG^{'}e^{A^{'}(T-t)}dt\right)^{-1}\\
=& \left( e^{A^{'}T}\right)^{-1}\left(\int_0^T e^{-At}GG^{'}e^{-A^{'}t}dt \right)^{-1}\left(e^{AT}\right)^{-1}\nonumber\\
=&\left( e^{A^{'}T}\right)^{-1} \tilde{\Sigma}_{\eta}(A, G, T)^{-1}\left(e^{AT}\right)^{-1},
\end{align}
and
\begin{align}
W(A, b_m, T) =& \left(e^{AT}\right)\int_0^Te^{-At}b_mb_m^{'}e^{-A^{'}t}dt\left( e^{A^{'}T}\right)\\
=&\left(e^{AT}\right)\tilde{W}(A, b_m, T)\left( e^{A^{'}T}\right).
\end{align}
}
Consequently, the expression in (\ref{MIObjective2}) is equivalent to:
{\small
\begin{align}
\ln\Biggl(\Biggl|I_{n\times n} + \tilde{\Sigma}_{\eta}^{-1}(A, G, T) \sum_{m=1}^p \tilde{W}(A, b_m, T)\Biggr|\Biggr)\label{MIObjective3},
\end{align}
}
\QEDA
\end{proof}
When the system is unstable (i.e.\ all eigenvalues of $A$ are positive), then the controllability Gramians in (\ref{GramianNoiseTilde}) and (\ref{GramianControlTilde}) for $T\rightarrow \infty$ can be found analytically by solving the corresponding continuous time Lyapunov equations:
\begin{align}
A\tilde{\Sigma}_{\eta}(A, G, \infty) + \tilde{\Sigma}_{\eta}(A, G, \infty)A^{'} =& GG^{'}, \label{GramianNoiseTilde}\\
\forall m:A\tilde{W}(A, b_m, \infty) + \tilde{W}(A, b_m, \infty)A^{'} =& b_mb_m^{'},\label{GramianControlTilde}
\end{align}
and, the 'water-filling' solution in (\ref{WFS1}) is computed with $W=\tilde{W}$, and $\Sigma_{n}=\tilde{\Sigma}_{n}$. 

\subsection{Infinitesimal time horizon}
For infinitesimal times, empowerment becomes a linear function of the horizon T, with a constant that depends on various parameters of the system. In particular, for $T\rightarrow 0$, empowerment vanishes again.

Particularly, the linearized noise covariance matrix in (\ref{totNoise}) appears as:
{\small
\begin{align}
\Sigma_{n}(T\ll 1) =& \left(\sigma_{\eta}Ce^{A0}GG^{'}e^{A^{'}0}C^{'}+\sigma_{\nu}\mathbf{I}\right)T\nonumber\\=&  \left(\sigma_{\eta}CGG^{'}C^{'}+\sigma_{\nu}\mathbf{I}\right)T,
\end{align}
}
which follows from the approximation of the integral in (\ref{NoiseGramian}) for $T\ll 1$.  Similarly, the control process covariance matrix (\ref{ControlGramian}) for $T\ll 1$ is:
{\small
\begin{align}\label{ControlGramian2}
\Sigma_u(T\ll 1) =& \sum_{m=1}^p\sum_{i=1}^n\sigma_{im} e^{A0}b_mg_{im}(0)e^{A^{'}0}b_m^{'}g_{im}(0)T^2\\
=&\left(\sum_{m=1}^p\sum_{i=1}^n\sigma_{im}b_mg_{im}(0)b_m^{'}g_{im}(0)\right)T^2,
\end{align}
}
which is quadratic in time. Consequently, the capacity for $T\ll 1$ appears as
{\small
\begin{align}
C(T\ll 1) =& \ln\left|I_{n\times n} + \left(\Sigma_{n}(T\ll 1)\right)^{-1}\Sigma_u(T\ll 1)\right|\\
=& \ln\left|I_{n\times n} + M_{n\times n}T\right| \approx \mathbf{Tr}\left(M\right)T,
\end{align}
}
where $\scriptstyle{M=\left(\sigma_{\eta}CGG^{'}C^{'}+\sigma_{\nu}\mathbf{I}\right)^{-1}\left(\sum\limits_{m=1}^p\sum\limits_{i=1}^n\sigma_{im}b_mg_{im}(0)b_m^{'}g_{im}(0)\right)}$ is a constant matrix.
}
\section{Simulations}\label{sec:Simul} 
Using the developed method, we demonstrate by computer simulations the relationship between empowerment and the intrinsic properties of linear dynamic systems. 

In the simulation we run the system for the fixed duration of $T$ seconds, however, we apply control only for the first $t$ seconds where $t\in[0,T]\mathrm{sec}$. That way, one can identify at which time scales a control signal has effect on the observed behaviour. 


The systems studied here were specified via their pole-zero map and are shown in Fig.\ref{FigCTPplan}.  Each system, (marked by a distinct color), is given by three pairs of the conjugate poles, $\scriptstyle{(\{p_1, p_1^{\dagger}\}, \{p_2, p_2^{\dagger}\}, \{p_3, p_3^{\dagger}\})}$, and the zero, $z$, at the origin. The poles of each system are located at increasing distances from the imaginary axis, which is equivalent to decreasing the {\it time constant}, $\tau$, of the system. The time constant is an intrinsic property of a dynamic system, which reflects the speed and the amplitude of the input signal propagation through the system. Increasing distances from the imaginary axis is equivalent to increasing the time constant, $\tau$, of the system.
\begin{figure}[h!]
\centering 
\includegraphics[scale=0.5]{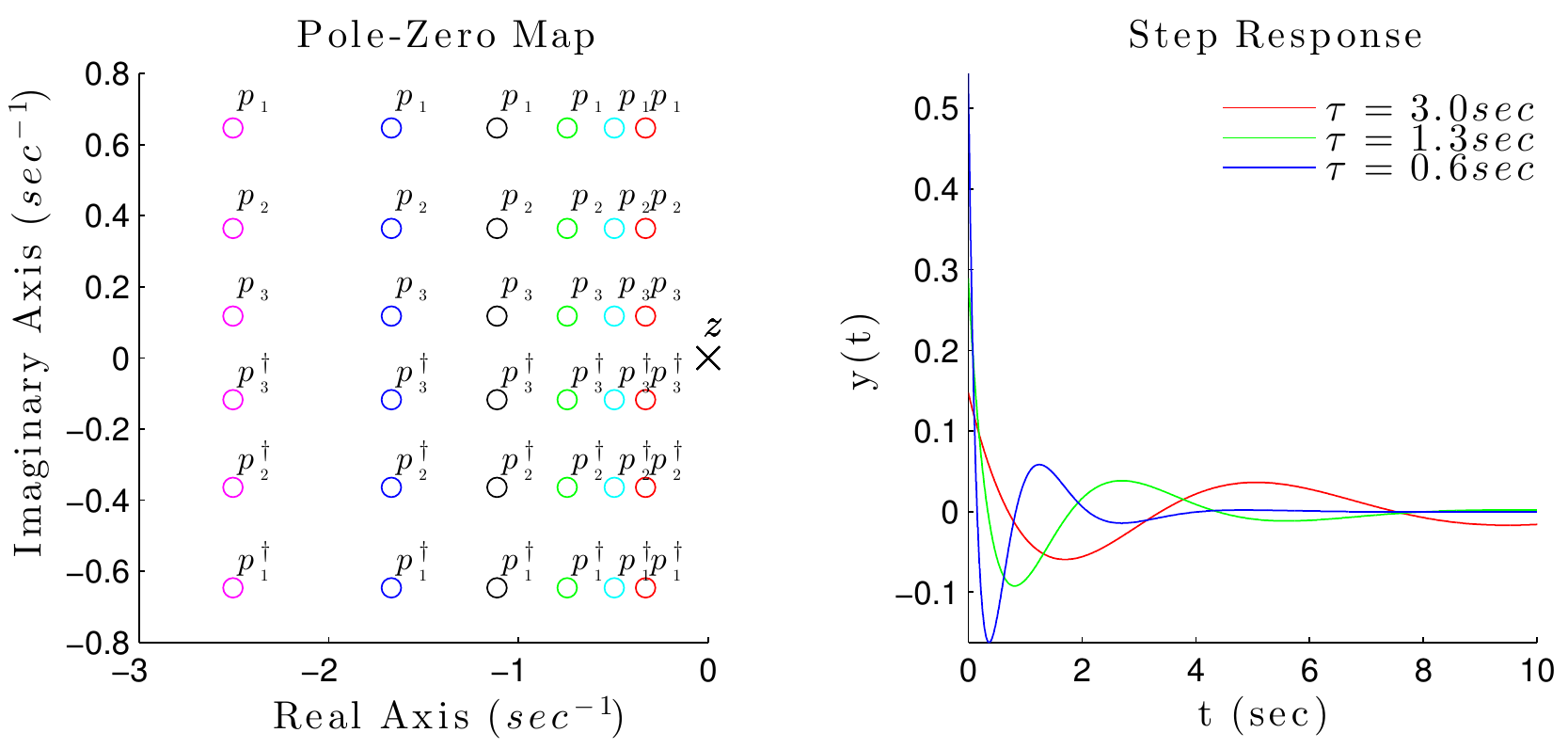}
\caption{{\bf Left plot}: Zero-pole maps of the continuous-time dynamic systems of order $6$. {\bf Right plot}: the step response of the systems in red, green, and  blue with the corresponding time constants, $\tau=3sec$, $\tau=1.3sec$, and $\tau=0.6sec$, respectively.}
\label{FigCTPplan}
\end{figure} 

For the systems shown in Fig.\ref{FigCTPplan}, we compute the empowerment landscapes, $\scriptstyle{C(P, T)}$, for $\scriptstyle{T\in[0,10]}$ seconds, and $\scriptstyle{P\in[0,10]}$ Watts, which are shown at Fig.\ref{CPT}. It is seen that empowerment is a monotonic function of the power for a given time, and it has a distinct maximum, (bright red region in the picture) as function of $T$. The time of empowerment maximum decreases with the decrease of the time constant, $\tau$. 

The physical interpretation of this effect is that only a limited time window of the input control signal contributes to the input-output mutual information. While, beyond this time window there is no effective contribution to the capacity. Otherwise, the capacity would have grown without limit with the length of the control signal. This time window depends on the correlation decay time of linear control system.
\begin{figure}[h!]
\centering 
\includegraphics[scale=0.45]{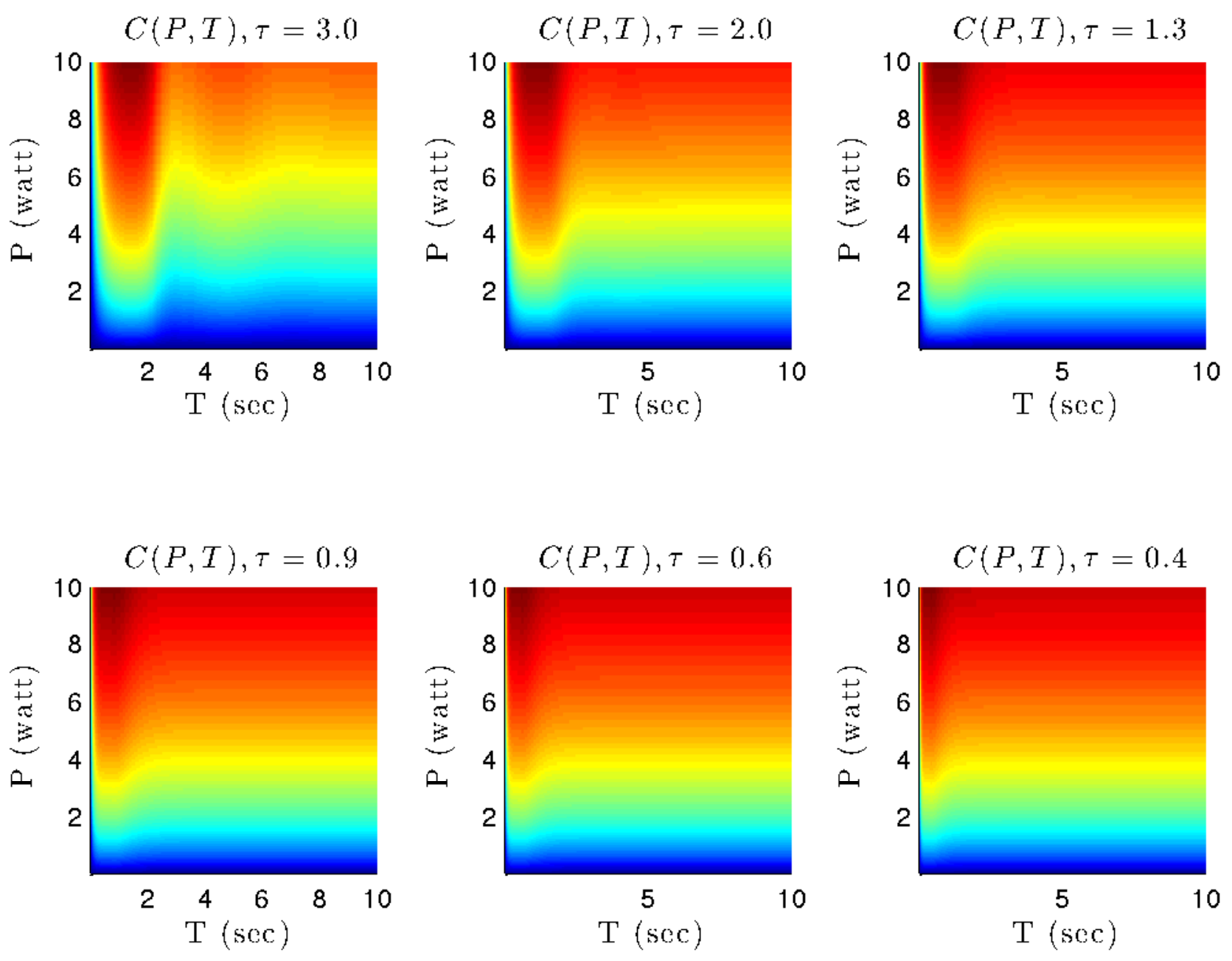}
\caption{The empowerment landscapes for the systems shown at Fig.\ref{FigCTPplan}. The blue and the red colors correspond to the low, ($C\rightarrow 0$ bits), and the high, ($C\sim 3$ bits), values of the empowerment, respectively.}
\label{CPT}
\end{figure} 
To zoom in into the empowerment landscape we consider the horizontal and the vertical slices of ${C(P, T)}$ at ${P=5}W$, and ${T=5}sec$, respectively. As seen in Fig.\ref{Slice}, the graphs of ${C(P, T)}$ become more similar to each other when the time constant, $\tau$, decreases. To elucidate this effect we parametrized empowerment by the time constant, $\tau$, and simulated ${C_{\tau}(P, T)}$ for a range of the time constant. 
\begin{figure}[h!]
\centering 
\includegraphics[scale=0.5]{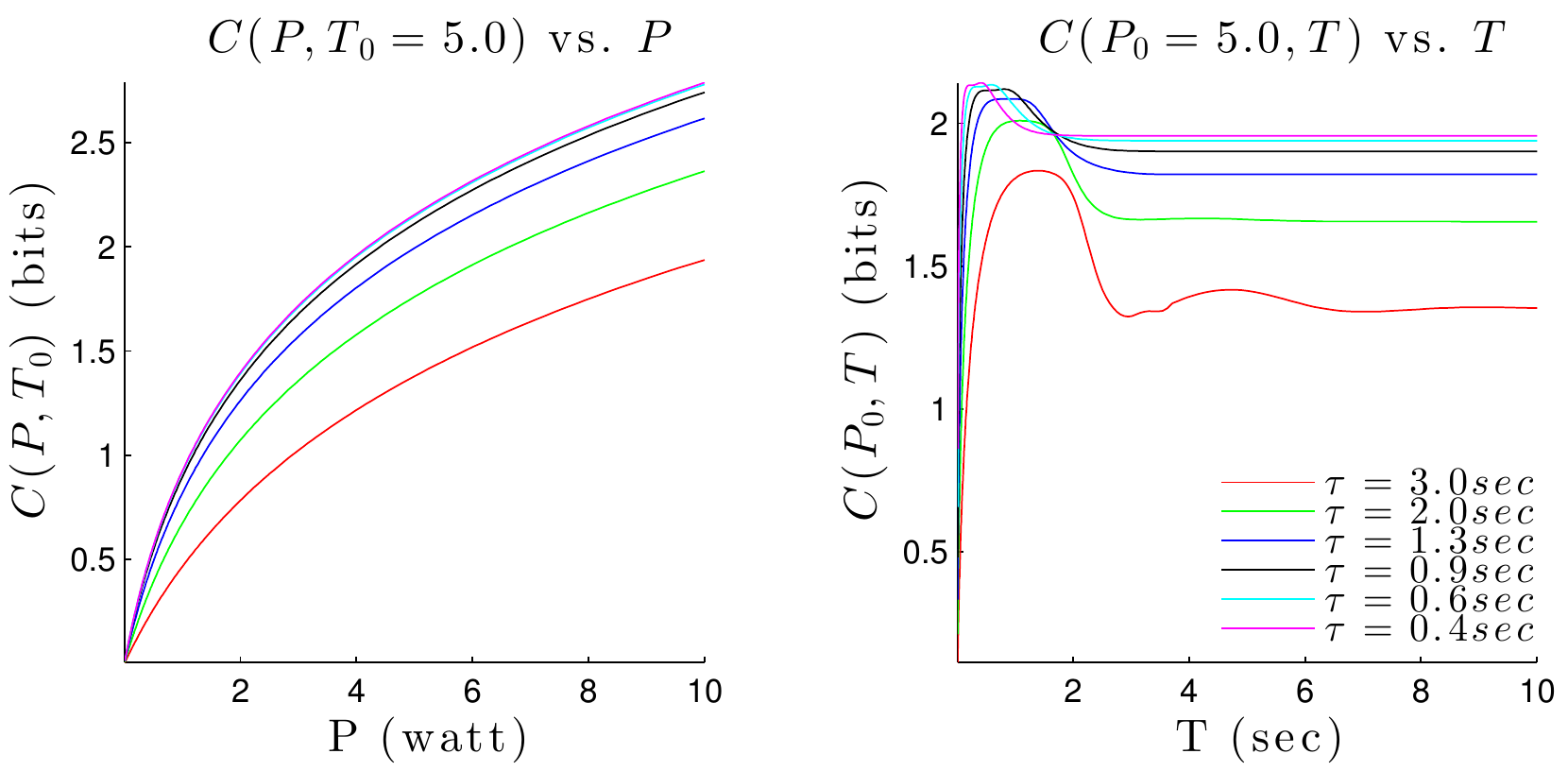}
\caption{The horizontal and the vertical slices of $C(P, T)$ from Fig.\ref{CPT} at $T=5sec$ and $P=5W$, respectively.}
\label{Slice}
\end{figure} 

As shown in Fig.\ref{CPT_TauConverge}, ${C_{\tau}(P, T)}$ converges to the finite value, ${C_{\tau}(T)\sim1.98}{bits}$ for  $\tau\rightarrow 0$. This behaviour, (the convergence to a finite capacity for decreasing $\tau$), is qualitatively similar for different orders of dynamic systems and for different power, $P$. It is notable that the same value is achieved when $\tau \rightarrow  0$, independently of the overall control time. Mathematically, the ratio of the noise Gramian to the control Gramian in (\ref{MIObjective2}) converges to a constant for $\tau\rightarrow 0$:
{\small
\begin{align}
\underset{\tau \rightarrow 0}{\lim}\;\Sigma_{\eta}^{-1}(A_{\tau}, G, T) \sum_{m=1}^p W(A_{\tau}, b_m, T)=\sigma_{\eta}\cdot(GG^{'})^{-1}\cdot \sum_{m=1}^pb_mb_m^{'}.
\end{align}
}

The physical interpretation of this effect is as following. The larger the distance of a pole from the imaginary axis is, the broader the frequency response, and simultaneously, the lower the gain at a corresponding frequency is. Moving the poles away from the imaginary axis is equivalent to increasing the system damping (or, equivalently, decreasing the time constant, $\tau$). There is a frequency window (as opposite to the abovementioned time window) in which the input control signal that can contribute to the value of empowerment, when the damping of the system is increasing. Empowerment achieves a finite value, as demonstrated in Fig.\ref{CPT_TauConverge}, even if $\tau\rightarrow 0$. 

\begin{figure}[h!]
\centering 
\includegraphics[scale=0.35]{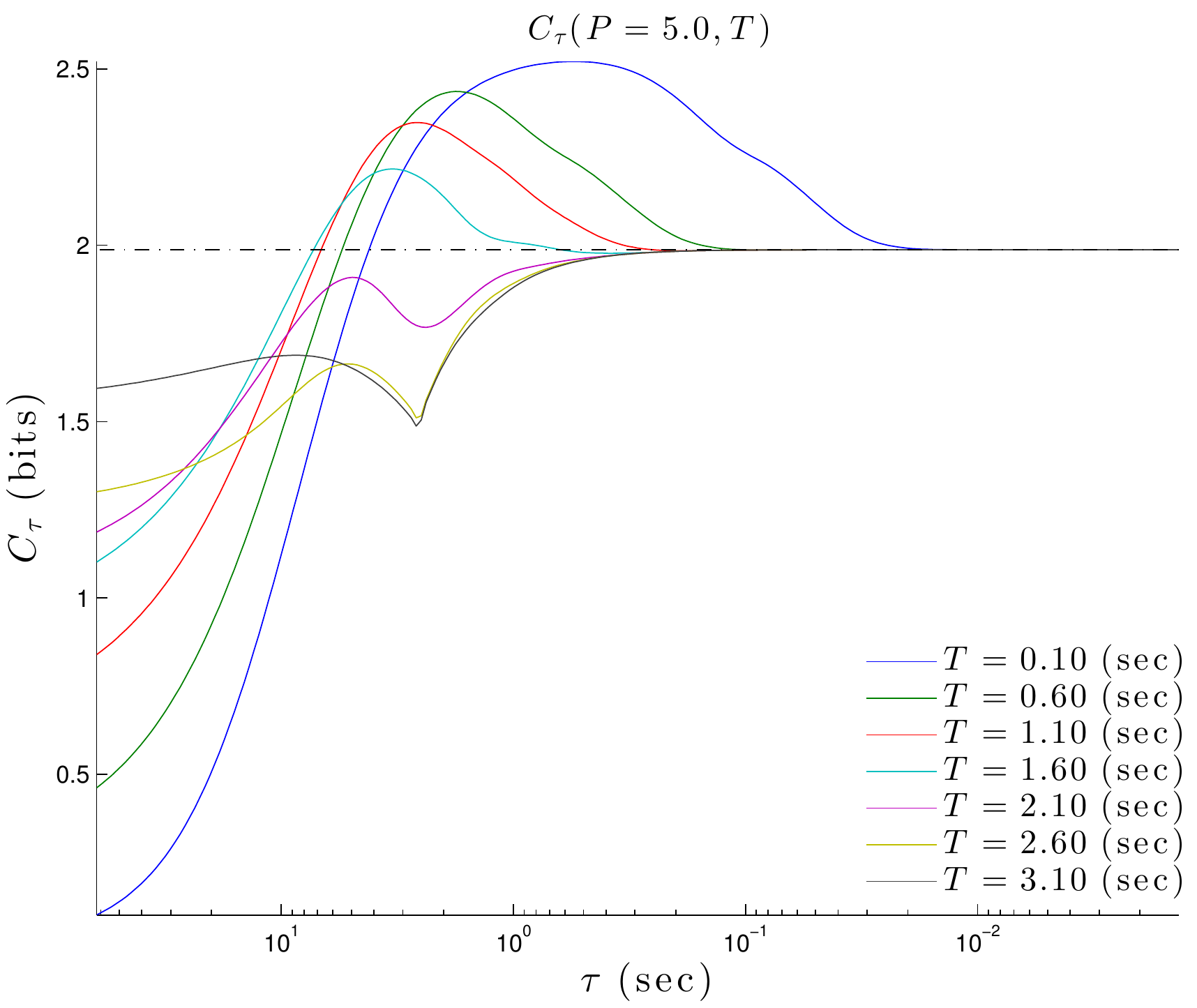}
\caption{Plot of $C_{\tau}(P, T)$. The x-axis is actually showing smaller $\tau$ as progressing to the right, ($\tau\rightarrow 0$). The dashed-line represents $C_{\tau\rightarrow 0}(P, T)$}. 
\label{CPT_TauConverge}
\end{figure} 

\section{Summary}\label{sec:Summ}
In this work we explore the information channel between the control signal to the future outputs of a stochastic linear control system. We derive an efficient method for computing the maximal mutual information, the capacity, between the control signal and the system future output. This can be viewed as solving the dual question to the paper by Tatikonda \& Mitter on control under communication constraints \cite{InfoContr1}. 

The capacity of the information channel between the control and the future output, known as empowerment, was shown in earlier work to be useful in substituting problem-specific utilities in established AI \& Control benchmarks,  \cite{Empowerment1,  Empowerment3, Empowerment4, Empowerment6,   Empowerment9}.

We develop an efficient method to estimate the value of empowerment in continuous-time linear stochastic control systems. We explore empowerment in partially observable systems, where the distinguishable control (actuator) signals are evaluated from the imperfect sensor signal. Empowerment can be - among other - interpreted as how well the control signals are distinguishable via the imperfect sensors at a future time; or, equivalently, how many distinguishable (with respect to the imperfect sensor) futures can be selected by the control plant. 

The results of this work elucidate the effect of the finite length of the input history, which holds the relevant information about the future. This is demonstrated by the constant capacity for growing $T$, as seen in Fig.\ref{Slice}. We demonstrate this using the example of stochastic dynamic systems, where we computed explicitly the control capacity and demonstrated the method by  computer simulations. 

The results of this work are derived for linear systems. However, it makes the approximative computation of this channel capacity accessible also for nonlinear systems via linearization. Specifically, non-linear dynamics can be approximated locally, comprising a linear time-variant, (LTV), dynamics. However, in the case of LTV dynamics, one can not in general derive a closed-form analytic expression, as in the case of a LTI system. 

The established method for the computation of empowerment in continuous-time linear control systems provides important insights to the understanding of the information processing in linear control systems. Particularly, we show that there is a limited time of control signal history which contributes effectively to the input-output mutual information. This time window defines an maximal length of the control signal, where the control affect the future output of the system.

The existence of a finite-sized effective time window for achieving the maximal capacity has a broader interpretation beyond the engineering interpretation for the systems: Optimized systems need, in given settings, only a limited history in order to make valuable predictions, required for the survival of an organism or device and/or for the accumulation of a value, \cite{InfoDecisions,InfoDecisions2,StillThermoPred}. In this work we demonstrate the {\it effect of the length the past window} on the example of stochastic dynamic system, and compute the capacity. 

We show that increasing the damping (or, equivalently, decreasing the time constant, $\tau$,) of a continuous-time linear control system does not result in zero empowerment. The reason for a non-zero value of empowerment for $\tau\rightarrow 0$ is the dependency of the mutual information on the Gramian ratio, (\ref{MIObjective2}). The ratio achieves a finite value for $\tau\rightarrow 0$.   
Technically, we show that the value of empowerment in continuous-time linear control systems can be computed without having to derive explicitly the function set of the Gaussian control process expansion, $\{g_i(t)\}_i,\;\;\forall t\in[0,T]$. Rather, it is sufficient to derive the '$z$-vectors', $\{z_i\}$, by the spectral decomposition of the controllability Gramian. In particular, there may exist different sets of the expansion functions satisfying the set of integral equations in (\ref{suffCond1}) and (\ref{suffCond2}). These degrees-of-freedom may be exploited for designing dynamic control systems with different energetic and/or complexity constraints, on the expansion function of the control signal. The optimal variances, $\{\sigma_i\}_i$, of the bi-orthonormal expansion of the Gaussian process (\ref{controlBiExt}) can be found efficiently by the iterative water-filling procedure.
The following direction of this work is a design of algorithms based on the empowerment method for the control of continuous-time linear systems and its extension to approximate empowerment in non-linear systems.
  
\section*{Acknowledgment}
The authors thank ICRI-CI, the Israel Science Foundation, and the Gatsby Charitable Foundation. The second author was supported in part by the EC Horizon 2020 H2020-641321 socSMCs FET Proactive project and the H2020-645141 WiMUST ICT-23-2014 Robotics project (Grant agreement no: 645141, Strategic objective: H2020 - ICT-23-2014 - Robotics).

\begin{appendices}
\section{Proof of Equation \ref{zVecEq}}\label{App:AppendixB}
\begin{proof}
Find $\nu_{mik}$ for all $k$ by $a)$ multiplying (\ref{gipOpt}) by $g_{km}(t)$, $b)$ integrating the result over $t\in[0,T]$, and $c)$ applying the orthonormality constraint, as following: 
{\small
\begin{align}
a)\,\sigma_{im}\mathbf{Tr}  \Bigl\{\Sigma_y^{-1} Cz_{im} w^{'}_{A, b_m}(t)C^{'}&g_{km}(t)\Bigr\}+\nonumber\\ 
&+\sum_{j}\nu_{mij}g_{jm}(t)g_{km}(t)=0\nonumber,
\end{align}
\vspace{-0.75cm}
\begin{align}
b)\,\sigma_{im}\mathbf{Tr} \Biggl\{\Sigma_y^{-1} Cz_{im} & {z}^{'}_{km}C^{'}\Biggr\}+\nonumber\\
&+\sum_{j}\nu_{mij}\int_0^Tg_{km}(t)g_{jm}(t)dt=0,\nonumber\\
c)\,\sigma_{im}\mathbf{Tr}\Biggl\{\Sigma_y^{-1}Cz_{im}& z^{'}_{km}C^{'}\Biggr\} +\nu_{mik} = 0.\label{nupijOpt}
\end{align}
}
Substitute $\nu_{mik}$ in (\ref{nupijOpt}) to (\ref{gipOpt}):  
{\small
\begin{align}
\sigma_{im}\mathbf{Tr}\bigl\{\Sigma_y^{-1} Cz_{im}& w^{'}_{A, b_m}(t)C^{'}\bigr\}-\nonumber\\
&-\sum_{j} \sigma_{im}\mathbf{Tr}\left\{\Sigma_y^{-1}C z_{im} z^{'}_{jm}C^{'}\right\}g_{jm}(t)=0,\nonumber
\end{align}
}
which, due to the trace circularity, is equal to the following expression:
{\small
\begin{align}
\sigma_{im}w^{'}_{A, b_m}(t)C^{'}\Sigma_y^{-1} Cz_{im} -\sigma_{im}\sum_{j}g_{jm}(t) z^{'}_{jm}C^{'}\Sigma_y^{-1} Cz_{im}=0,\label{removeTrace}
\end{align}
}
Finally, multiplying (\ref{removeTrace}) by $w_{A, b_m}(t)$ and integrating over $t\in[0,T]$ we get:  
{\small
\begin{align}
\sigma_{im}\Biggl(\int_0^Tw_{A, b_m}(t)&w^{'}_{A, b_m}(t)dt\Biggr)C^{'}\Sigma_y^{-1}Cz_{im} \nonumber\\
&-\sigma_{im}\sum_{j} z_{jm} z^{'}_{jm}C^{'}\Sigma_y^{-1}Cz_{im}=0,\nonumber
\end{align}
}
which is equal to (\ref{zVecEq}), using the notations in (\ref{GramianNotation}):
{\small
\begin{align}\label{uniqSolLemma}
\!\!\!\!\!\forall i,m:\sigma_{im}\!\!\left(\!\!W(A, b_m, T)\!-\!\sum_{j} z_{jm} z^{'}_{jm}\!\!\right)\!\!C^{'}\Sigma_y^{-1} Cz_{im}\!\!=0.\nonumber
\end{align}
}$\hfill\blacksquare$  
\end{proof}
\end{appendices}
\ifCLASSOPTIONcaptionsoff
  \newpage
\fi
\bibliographystyle{IEEEtran}

\end{document}